\documentclass[conference, onecolumn]{IEEEtran}

\IEEEoverridecommandlockouts

\usepackage{graphicx}
\usepackage{times}
\usepackage{amsmath}
\usepackage{amsfonts}
\usepackage{amssymb}
\usepackage{stfloats}
\usepackage{setspace}
\usepackage{algorithm}
\usepackage{algorithmic}
\usepackage{morefloats}
\usepackage{color}
\usepackage{subfigure}
\usepackage{fmtcount}
\usepackage{url,cite}
\usepackage{bm}
\usepackage{paralist}

\newtheorem{proposition}{Proposition}
\newtheorem{lemma}{Lemma}
\newtheorem{corollary}{Corollary}

\allowdisplaybreaks

\newtheorem{prop}{Proposition}[section]

\newtheorem{cor}{Corollary}[section]

\newtheorem{lm}{Lemma}[section]

\newtheorem{thm}{Theorem}[section]

\newcommand{\bthm}{\begin{thm}}
\newcommand{\ethm}{\end{thm}}

\newcommand{\bcor}{\begin{cor}}
\newcommand{\ecor}{\end{cor}}
\newcommand{\bprop}{\begin{prop}}
\newcommand{\eprop}{\end{prop}}
\newcommand{\blm}{\begin{lm}}
\newcommand{\elm}{\end{lm}}
\newcommand{\beq}{\begin{equation}}
\newcommand{\eeq}{\end{equation}}
\newcommand{\ber}{\begin{eqnarray}}
\newcommand{\eer}{\end{eqnarray}}

\newenvironment{proof1}{\begin{trivlist}\item[]{\bf Proof:\hspace{2mm}}}{\hfill$\blackbox$\end{trivlist}}

\newcommand{\blackbox}{\vrule height7pt width5pt depth1pt}

\newcommand{\bit}{\begin{itemize}}
\newcommand{\eit}{\end{itemize}}
\newcommand{\ben}{\begin{enumerate}}
\newcommand{\een}{\end{enumerate}}
\newcommand{\bdesc}{\begin{description}}
\newcommand{\edesc}{\end{description}}
\newcommand{\beqarrn}{\begin{eqnarray*}}
\newcommand{\eeqarrn}{\end{eqnarray*}}
\newenvironment{proofof}[1]{\begin{trivlist}\item[]{\bf Proof of #1:\hspace{2mm}
}}{\hfill\blackbox\end{trivlist}}
\newcommand{\bproofof}{\begin{proofof}}
\newcommand{\eproofof}{\end{proofof}}
\newenvironment{rem}{\begin{trivlist}\item[]{\bf
Remark:}\hspace{4mm}}{\end{trivlist}}
\newcommand{\brem}{\begin{rem}}
\newcommand{\erem}{\end{rem}}
\newenvironment{rems}{\begin{trivlist}\item[]{\bf
Remarks}\begin{itemize}}{\end{itemize}\end{trivlist}}
\newcommand{\brems}{\begin{rems}}
\newcommand{\erems}{\end{rems}}
\newtheorem{fact}{Fact}
\newcommand{\bfact}{\begin{fact}}
\newcommand{\efact}{\end{fact}}
\newtheorem{examp}{Example}
\newcommand{\bexamp}{\begin{examp}\rm}
\newcommand{\eexamp}{\end{examp}}
\newtheorem{defn}{Definition}[section]
\newcommand{\bdefn}{\begin{defn}\rm}
\newcommand{\edefn}{\end{defn}}

\newtheorem{prob}{Problem}
\newcommand{\bprob}{\begin{prob}}
\newcommand{\eprob}{\end{prob}}

\newcommand{\bvtm}{\begin{verbatim}}
\newcommand{\bfig}{\begin{figure}}
\newcommand{\efig}{\end{figure}}
\newcommand{\bcen}{\begin{center}}
\newcommand{\ecen}{\end{center}}

\long\def\comment#1{}

\def \n2{{N_0 \over 2}}

\def \h5{\hspace{0.5in}}

\begin{document}

\title {Exploiting Channel Correlation and PU Traffic Memory for Opportunistic Spectrum Scheduling}

\author {
{Shanshan Wang,
Sugumar Murugesan
and Junshan Zhang}
\\
\centering{School of Electrical, Computer and Energy Engineering,
Arizona State University, Tempe, AZ, 85287}}

\maketitle

\thispagestyle{plain}
\pagestyle{plain}

\begin{abstract}
We consider a cognitive radio network with multiple primary users (PUs) and one secondary user (SU), where
a spectrum server is utilized for spectrum sensing and scheduling the SU to transmit over one of the PU channels opportunistically.
One practical yet challenging scenario
is when \textit{both} the PU occupancy and the channel fading vary over time and exhibit temporal correlations.
Little work has been done for exploiting  such temporal memory in  the channel fading and the PU occupancy simultaneously for opportunistic spectrum scheduling.
A main goal of this work is to understand the intricate tradeoffs resulting from the interactions of the two sets of system states - the channel fading and the PU occupancy,
by casting the problem as a partially observable Markov decision process.
We first show that a simple greedy policy is optimal in some special cases. To build a clear understanding of the tradeoffs,
we then introduce a full-observation genie-aided system, where the spectrum server collects channel fading states from all PU channels.  The genie-aided system is used to decompose the tradeoffs in the original system into multiple tiers, which are examined progressively.
Numerical examples indicate
that the optimal scheduler in the original system, with observation on the scheduled channel only, achieves a performance very close to the genie-aided system.
Further, as expected, the optimal policy in the original system significantly outperforms randomized scheduling, pointing to the merit of exploiting the temporal correlation structure in both channel fading and PU occupancy.
\end{abstract}


\section{Introduction}

Over the past decade, cognitive radio (CR) has been identified as one promising solution to ease the ``spectrum scarcity'' associated with the traditional \textit{static} spectrum allocation \cite{Survey:Qing1, NeXt, Fette06, thesis:CR}. Going beyond the fixed and licensed spectrum allocation, a secondary user (SU) can opportunistically access the spectrum owned by the primary users (PUs) in a CR network. This paradigm shift from static to dynamic spectrum allocation has been shown to bring significant improvement in the spectrum utilization, and hence the system's overall performance.

A fundamental principle enabling the cognitive capability is built upon the SU's dynamic adaptation of its operation parameters (such as power, frequency, etc.), according to the environmental variations over time. One such variation  is the channel fading. Often times an \textit{i.i.d.} flat fading model is used in abstracting fading channels, which fails to capture the temporal channel memory observed in realistic scenarios~\cite{fading}. An alternative model, namely the
Gilbert-Elliot (GE) model~\cite{Gilbert}, has been widely used (see, e.g., \cite{myopic:Qing, sugumar, KKar_infocom07}) to capture the temporal correlation in the fading process.
Specifically, the GE model uses a first-order Markov chain with two states: one representing a ``good'' channel where the user experiences error-free transmissions, and the other representing a ``bad'' channel with unsuccessful transmissions.

Another variation to be considered is the PU's activity on the channels.
Note that in a CR network,
the SUs have a strictly lower priority in the spectrum usage, and can only access the channels when the PUs are absent~\cite{Survey:Qing1, NeXt, Fette06, thesis:CR}. This unique spectrum usage structure necessitates the inclusion of the channel's PU occupancy state in determining the channel's accessability by an SU.

In many of the existing works (see, e.g., \cite{myopic:Qing, sugumar, LAJohnston_TW06, QZhao_selfsimilar, ASwami09}), only one set of the system states -- either the channel fading, or the PU occupancy
-- has been taken into consideration in developing spectrum access strategies by the SU. In this work, we take a step forward, and
explore the utility of \textit{both} the states for opportunistic channel access. Specifically, we consider a CR network consisting of multiple PU channels and one SU opportunistically accessing one of the PU channels at a time. A spectrum server is utilized to periodically schedule the SU to one of the channels for transmission.
Worth noting is that the usage of the spectrum server is consistent with the recent FCC ruling on the use of a spectrum database in CR network operations~\cite{FCC_database}.
Further, the spectrum server facilitates spectrum management, and enhances the scalability of the network~\cite{Akyildiz_survey}.

Dynamic spectrum access in the presence of the temporal variations
can be cast as a sequential control problem.
We formulate this sequential control problem as a partially observable Markov decision process~\cite{POMDP}.
In this context, the spectrum server makes scheduling decisions in terms of allocating a PU channel to the SU, based on the channel's PU occupancy state and
fading state.
We model the channel fading by using a two-state first-order Markov chain, i.e., the GE model. On the other hand,
since the PU activity may possess a long temporal memory (see, e.g., \cite{onoffmodel2, DWillkomm_Dyspan08}),
we develop an ``age'' model to capture the temporal correlation structure of the PU occupancy state.

Building on the above model, we examine the intricate tradeoffs resulting from the dynamic interaction of the system states.
Our main contributions can be summarized as follows:
\begin{itemize}
  \item We study opportunistic spectrum scheduling by exploiting the temporal correlation structure
   in \textit{both} the channel fading and the PU occupancy states. This, to the best of our knowledge, has not been addressed systematically in the literature so far.

  \item We show that the optimal scheduling involves a \textit{multi-tier} ``exploitation vs. exploration'' tradeoff. For certain special cases,
  we establish the optimality of a simple greedy policy, and examine the intricacy of the fundamental tradeoffs.

  \item To gain a better understanding of the tradeoffs for the general case, we introduce a full-observation genie-aided system, where the spectrum server collects channel fading states from all the PU channels. Using the genie-aided system, we decompose the multiple tiers of the tradeoffs, and examine them progressively.
\end{itemize}

The rest of the paper is organized as follows. Section~\ref{sec:probform} introduces the basic setting and problem formulation in detail. In Section~\ref{sec:tradeoffs}, we identify the fundamental tradeoffs and illustrate them via special cases. Section~\ref{sec:tradeoff_ga} further examines the tradeoffs by developing a genie-aided system that isolates the impact of channel fading and PU occupancy on the optimal reward. In Section~\ref{sec:numresult}, numerical results are presented where we evaluate and compare the performance of the optimal policy in the original system with baseline cases. We also study the impact of the memory in the channel fading and PU occupancy on the relative performances of various baseline cases. This is followed by concluding remarks in Section~\ref{sec:conclusion}.

\section{Problem Formulation} \label{sec:probform}

\subsection{Basic Setting} \label{subsec:sysmodel}
We consider a CR network with one SU and $N$ PUs\footnote{Each user is assumed to be a pair of transmitter and receiver.}. Each PU is licensed to one of $N$ independent channels, henceforth identified as PU channels.
A PU generates packets according to a stationary process,
transmits over its channel if there are backlogged packets, and leaves upon the completion of the transmissions. The PU traffic activity
is assumed to be identical and independent across channels.

The SU, on the other hand, is backlogged with packets and
opportunistically transmits these packets over
the PU channels with the help of a spectrum server.
Time is divided into two timescales: mini-slots and the control slots each constituting $K$ mini-slots, as illustrated in Fig.~\ref{fig:two_time_scale}.
The length of each mini-slot is normalized to fit the transmission of one data packet of the PU or the SU.
At the beginning of each control slot, the spectrum server schedules the SU to the ``best'' PU channel that is expected to yield the highest average throughput for the SU.
The SU then transmits packets in the scheduled channel, until it detects the return of a PU\footnote{This can be accomplished by incorporating collision detection by the SU at the mini-slot timescale. We also assume that PU arrivals coincide with the mini-slot boundaries.}.
Upon such an event, the SU suspends transmissions until the beginning of the next control slot, when the spectrum server re-schedules the SU to a PU channel based on most recent observations. At the end of each mini-slot when the SU transmitted a packet, it sends accurate feedback on the channel fading state (of the PU channel, as seen by the SU) corresponding to that mini-slot, to the spectrum server.
The spectrum server uses this channel fading feedback, the PU traffic observations, along with the memory inherent in these processes to perform informed scheduling decisions at the beginning of the next control slot. We discuss the system model and the scheduling problem formulation in more detail in the following.

\begin{figure}[tbh]
    \begin{center}
        \includegraphics [width = 0.35\textwidth] {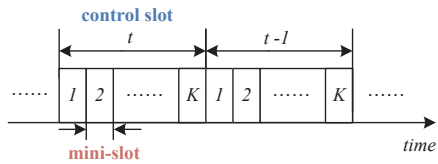}
        \caption{A sketch of the two timescale model.}
        \label{fig:two_time_scale}
    \end{center}
\end{figure}

\subsection{Problem Formulation} \label{subsec:probformulation}
The opportunistic spectrum access at hand can be viewed as a sequential control problem, which we formulate as a partially observable Markov decision process. In the following, we introduce and elaborate the entities involved in the formulation.

\textit{Channel occupancy:} The usage pattern on each of the PU channels can be modeled as an ON-OFF process at the mini-slot timescale, with ON denoting the \textit{busy} state where the PU transmits data over the channel, and OFF the \textit{idle} state where the PU is absent.
Channel occupancy is the idle or busy state of the PU channels. Let $o_{t,k}(n)$ be a binary random variable, denoting whether PU channel $n$, for $n\in\{1,\ldots,N\}$, is idle ($o_{t,k}(n)=0$), or not ($o_{t,k}(n)=1$), in the $k$th mini-slot of control slot $t$. The corresponding idle probability is denoted by $\pi_{t,k}^o(n) \triangleq \textrm{Pr}(o_{t,k}(n)=0)$.

\textit{Idle/Busy age:} The PU traffic is temporally correlated, i.e., the current occupancy state on each of the channel depends on the history of the channel occupancy. We introduce the notion of ``age,'' defined as follows, to characterize the occupancy history:

\textit{The \textbf{age} of a PU channel is the number of consecutive mini-slots immediately preceding the current mini-slot, during which the channel is in the same occupancy state as in the current mini-slot. The age is denoted as ``idle age'' if the channel is in idle state in the current mini-slot and ``busy age'' otherwise.} We use $x_t(n)$ to denote the age of channel $n$ at the beginning of control slot $t$.

As noted earlier, we assume long memory in the PU occupancy state. Specifically, with the definition of age in place, we adopt a family of functions monotonically decreasing in age, to denote the \textit{conditional probability that a channel will be idle (or busy), given that it has been idle (or busy) for $x \geq 1$ mini-slots}:
    \begin{eqnarray} \label{eq:PI_PB}
      P_I(x) &=& \frac{1}{x^u + C_I}, \nonumber\\
      P_B(x) &=& \frac{1}{x^u + C_B}, ~~u = 1,2,...,
    \end{eqnarray}
    where $C_I$ and $C_B$ are normalizing constants taking positive values.
    Our occupancy model essentially imposes the following realistic
    correlation  structures:
    \begin{inparaenum}[\itshape 1\upshape)]
    \item the occupancy memory weakens with time, i.e., the impact of past occupancy events on the current occupancy state diminishes since the said event happened;
    \item the conditional probability that the PU channel is busy or idle now, is purely a function of the length of time the channel has been in the most recent state, and is independent of the channel occupancy history before the time of the latest transition to the most recent state.
    \end{inparaenum}
    In sight of this, the quantities $P_I$ and $P_B$ defined in (\ref{eq:PI_PB}) are sufficient for capturing the temporal correlation in the channels' PU occupancy state.

\textit{Channel fading model:} At the end of each mini-slot after transmitting a packet, the SU measures the channel fading between its transmitter and receiver on the scheduled channel, and feeds back this information to the spectrum server. {Inspired by recent works~\cite{LAJohnston_TW06, sugumar, myopic:Qing}}, we capture the memory in the fading (of the PU channel)
between the SU's transmitter and receiver using a two-state, first-order Markov chain, with state variations occurring at the mini-slot timescale. The Markov chain model is \textit{i.i.d.} across the PU channels. Each state of the Markov chain corresponds to the degree of decodability of the data sent through the channel, where state 1 denotes full decodability and state 0 denotes zero decodability. Note that the states can also be interpreted as a quantized representation of the underlying channel fading, in the sense that state 1 corresponds to ``good'' channel fading, while state 0 corresponds to ``bad'' fading.
The probability transition matrix of this Markov chain is given as:
\begin{eqnarray} \label{eq:fading_MC}
  \textbf{P}:= \left[
    \begin{array}{cc}
        1 - r & r \\
        1 - p & p
    \end{array}
  \right],
\end{eqnarray}
where $p$ is the conditional probability that the channel fading is good, given that it was good in the previous mini-slot; and $r$ is the conditional probability that the channel fading is good, given that it was bad in the previous mini-slot. Throughout the paper, we will focus on the case when the fading channels are positively correlated, i.e., $p>r$.

\textit{Belief of channel fading state:} Denote by $\pi_{t,k}^s(n)$ the belief of channel fading state in the $k$th mini-slot of control slot $t$ on channel $n$. As is standard~\cite{POMDP, myopic:Qing}, the fading state belief is a sufficient statistic that characterizes the current channel fading state as perceived by the SU.
Further, let $f_{t,k}(a_t)$ be a binary random variable denoting the fading state feedback obtained at the end of the $k$th mini-slot in control slot $t$ on the scheduled channel $a_t$. Also, define $\mathrm{T}^L(\cdot)$, for $L \in\{ 0,1,\ldots\}$, as the $L$th step belief evolution operator, taking the form: for $\gamma \in (0,1)$,
\begin{equation}  \label{eq:TL_alpha}
    \mathrm{T}^L(\gamma) = \mathrm{T}(\mathrm{T}^{L-1}(\gamma)),
\end{equation}
with $\mathrm{T}^0(\gamma) = \gamma$ and $\mathrm{T}(\gamma) = \gamma p + (1-\gamma)r$.
Now, the update of the fading state belief is governed by the underlying Markov chain model, and any new information obtained on the channel fading, i.e.:
\begin{eqnarray} \label{eq:pis}
    \pi_{t,k+1}^s(n) =
    \left\{
    \begin{array}{ll}
        p, ~~~& a_t = n, f_{t,k}(a_t)=1, \\
        r, ~~~& a_t = n, f_{t,k}(a_t)=0, \\
        \mathrm{T}(\pi_{t,k}^s(n)),~~~& a_t \neq n.
    \end{array}
    \right.
\end{eqnarray}

\textit{Action space:} This refers to the set of channels that the scheduling decision is made from. The spectrum server selects channels only from those that are currently idle\footnote{This is a policy level constraint to protect the PU's priority in spectrum access.}, and the action space $\mathcal{A}_t$ in control slot $t$ can thus be written as:
\begin{equation}
  \mathcal{A}_t = \{n:o_{t,1}(n) = 0\}.
\end{equation}

\textit{State:} At the beginning of each control slot, the spectrum server makes the scheduling decision based on three factors: For each of the PU channels,
\begin{inparaenum}[\itshape 1\upshape)]
\item the idle/busy state at the moment;
\item the length of time the channel has been in the current occupancy state (i.e., age); and
\item the fading state belief value.
\end{inparaenum}
That is, the state of each PU channel $n$, is represented by a three dimensional vector: $S_t(n) = [o_{t,1}(n), x_t(n), \pi_{t,1}^s(n)]$. Accordingly, the state of the system at the beginning of current control slot $t$ is described by a $N \times 3$ matrix $\textbf{S}_t$:
\begin{eqnarray}
  \textbf{S}_t \hspace{-0.5mm}:= \hspace{-0.5mm}[S_t(1);\ldots;S_t(N)] \hspace{-1mm}=\hspace{-1mm} \left[\hspace{-1mm}
    \begin{array}{ccc}
        o_{t,1}(1) & x_t(1) & \pi_{t,1}^s(1) \\
        o_{t,1}(2) & x_t(2) & \pi_{t,1}^s(2) \\
        \vdots & \vdots & \vdots \\
        o_{t,1}(N) & x_t(N) & \pi_{t,1}^s(N)
    \end{array}
  \hspace{-1mm}\right].
\end{eqnarray}

\textit{Horizon:} The horizon
is the number of consecutive control slots over which scheduling is performed.
We index the control slots in a decreasing order with control slot $1$ being the end of the horizon\footnote{For the mini-slots, we use the conventional increasing time indexing.}.
Throughout the paper, we denote the length of the horizon by $m$, i.e., the scheduling process begins at control slot $m$.

\textit{Stationary scheduling policy:} A stationary scheduling policy $\mathcal{P}$
establishes a stationary mapping from the current state $\textbf{S}_t$ to an action $a_t$ in each control slot $t$.

\textit{Expected immediate reward:} The expected immediate reward is the reward accrued by the SU within the current control slot. Specifically, the SU collects one unit of reward in each mini-slot, if the channel is idle and has good channel fading (i.e., conditions that indicate successful transmission by SU). Since the scheduled channel must be idle in the first mini-slot of the current control slot, the expected immediate reward can be calculated as:
\begin{equation}\label{eq:Rt_general}
  R_t(\textbf{S}_t,a_t)=\sum_{k=2}^K \pi_{t,k}^o(a_t) \pi_{t,k}^s(a_t) + \pi_{t,1}^s(a_t).
\end{equation}

\textit{Total discounted reward:}
Given a scheduling policy $\mathcal{P}$, the total discounted reward, accumulated from the current control slot $t$, until the horizon, can be written as\footnote{In the subsequent sections, we may drop $\mathcal{P}$ and the tiers of expectation to simplify the notation.}
\begin{eqnarray}
  V_t(\textbf{S}_t; \mathcal{P}) = R_t(\textbf{S}_t,a_t) + \beta E_{{\bm{\pi}}_{t-1}^s} E_{\textbf{O}_{t-1}}  V_{t-1}\left(\textbf{S}_{t-1}; \mathcal{P}\right),
\end{eqnarray}
where $\beta\in(0,1)$ is the discount factor, facilitating relative weighing between the immediate and future rewards, and the expectation is taken with respect to fading state belief: $\bm{\pi}_{t-1}^s = \{\pi_{t-1,1}^s(1),\ldots,\pi_{t-1,1}^s(N)\}$, and PU occupancy: $\textbf{O}_{t-1} = \{o_{t-1,1}(1),\ldots,o_{t-1,1}(N)\}$.

\textit{Objective function:} The objective of the scheduling problem is to maximize the SU's throughput, i.e., SU's total discounted reward. A scheduling policy $\mathcal{P}^*$ is optimal if and only if the following optimality equation is satisfied:
\begin{equation}
  V^(\textbf{S}_t; \mathcal{P}^*) \hspace{-1mm} = \hspace{-1mm} \max_{a_t \in \mathcal{A}_t}\Bigg\{R_t(\textbf{S}_t, a_t)+ \beta E_{{\bm{\pi}}_{t-1}^s} E_{\textbf{O}_{t-1}} V^(\textbf{S}_{t-1}; \mathcal{P}^*)\Bigg\}.
\end{equation}
The function $V^*(\textbf{S}_t) := V^(\textbf{S}_t; \mathcal{P}^*)$ is the objective function of the scheduling problem.

\section{Fundamental Tradeoffs} \label{sec:tradeoffs}

The decision on opportunistic spectrum scheduling is made based on two sets of system states: the PU occupancy on the channel and the channel fading perceived by the SU. On one hand, PUs may return in the middle of a control slot and hinder further transmissions of the SU, leading to a decreased reward for the SU. The temporal memory resident in the PU occupancy suggests that the past history of channel's occupancy, measured by the age,
influences the occupancy state of the channel in the future.
On the other hand, the PU channels
may suffer from ``bad'' channel fading in the middle of a control slot, even if a PU does not return to hinder SU's transmissions..
Similar to the PU occupancy, the historic observation on the fading process would help determine the expected channel fading in the future. Note that by way of the channel feedback arrangement, an observation of a PU channel fading is made only when that channel is scheduled to the SU. Thus scheduling is inherently tied to channel fading learning.
Roughly speaking, to maximize the SU's reward, the spectrum server must schedule a channel such that the combination of the perceived channel occupancy and channel fading strikes a ``perfect'' balance between the immediate gains and channel learning for future gains. We discuss this intricate tradeoff in the following.

\subsection{Classic ``Exploitation vs. Exploration'' Tradeoff}
In the existing literature (e.g., \cite{myopic:Qing, sugumar, LAJohnston_TW06, QZhao_selfsimilar, ASwami09}), focus has been cast on considering only one of the factors: either channel fading or channel occupancy, along with the associated temporal correlation. The optimal decision is a mapping that best balances the tradeoff of ``exploitation'' and ``exploration'' on the single factor being considered. The exploitation side lets the scheduler choose the channel with the best perceived channel fading (or occupancy state) at the moment, corresponding to immediate gains; while the exploration side tends to favor the channel with the least learnt information so far, probing which can contribute to the overall understanding of the channel fading (or occupancy state) in the network, and thus better opportunistic scheduling decisions in the future.

\subsection{``Exploitation vs. Exploration'' Tradeoff in Dynamics of Both Channel Fading and PU Occupancy States}

In contrast to the existing works, we examine the tradeoffs when the temporal correlation in \textit{both} the channel fading and PU occupancy are considered.
While the classic tradeoff described above apparently exists, additional tradeoffs arise in our context due to the interactions between the two sets of system states.
In particular, the long temporal memory in the PU occupancy state adds a new layer to the tradeoffs inherent in the problem. For instance, note that the SU can only transmit on \textit{idle} channels. To carry out exploration, the channel being favored in the traditional sense, i.e., the one on which the least information is available, may no longer be the preferred choice if this channel is perceived to be unavailable (i.e., busy) for a prolonged duration in the near future. In other words, it may not be worth learning the channel as the SU cannot utilize the learned knowledge in the near future.

Fig.~\ref{fig:tradeoffs} is a pictorial illustration of the impact from the occupancy state on the SU's expected reward. The history of occupancy, represented by the idle ages $x_t(1)$ and $x_2(t)$, affects both the immediate and future rewards of the SU. Specifically, as the idle age increases, the temporal memory in the PU's occupancy pushes the channel to transit to busy sooner (i.e., time point $b$ comes earlier than $a$ in the figure). Therefore,
the average availability on the PU channel in the current control slot decreases, which leads to a smaller immediate reward for the SU.

Further, note that the latest mini-slot for which the spectrum server receives channel fading feedback is also the last mini-slot before the PU returns, i.e., time points $a$ and $b$ respectively for the two cases in Fig.~\ref{fig:tradeoffs}. The duration $d_1$ (likewise, $d_2$) in Fig.~\ref{fig:tradeoffs} is an indication of how ``fresh'' the channel fading information is for the scheduling decision at the beginning of the next control slot, i.e., $t-1$. With $d_1 < d_2$, channel feedback is more fresh in the former case, with a lower idle age $x_t(1)$.
Thus, age, through its effect on the \textit{freshness} of feedback, and the availability of the PU channels in the future slots, adds another layer to the tradeoffs, thereby influencing the optimal scheduling decision.

\begin{figure}[tbh]
    \begin{center}
        \includegraphics [width = 0.35\textwidth] {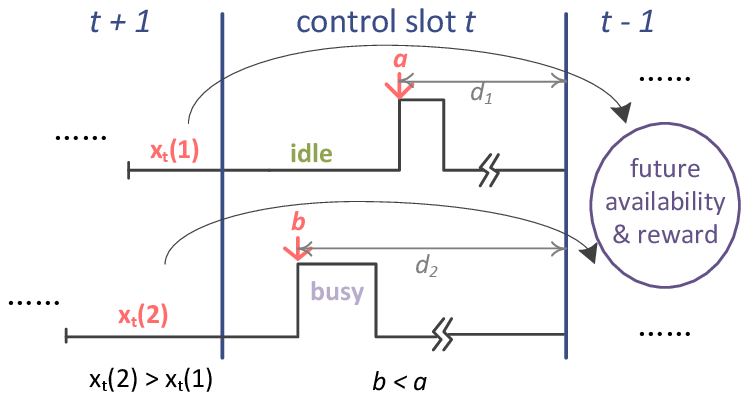}
        \caption{An illustration of the impact of age on the SU's reward.}
        \label{fig:tradeoffs}
    \end{center}
\end{figure}

To better perceive the intricate tradeoffs in the system, we proceed, in what follows, with a number of break-down results that aim at illustrating each tier of the tradeoffs progressively.

\subsection{Tradeoffs Inherent in Immediate Reward}

Consider the PU channel scheduled in the current control slot $t$. Let $k_0$ denote the latest mini-slot before the PU of the scheduled channel returns in the current control slot.
Clearly, $k_0$ is a random variable, taking values in $k_0 \in \{1,\ldots,K\}$. Let $p_z \triangleq \textrm{Pr}(k_0 = z)$. With $x$ denoting the idle age of the scheduled channel, $k_0$ is distributed as follows: For $K=2$:
\begin{equation*}
  p_1  = 1 - P_I(x+1),~~~
  p_2 = P_I(x+1),
\end{equation*}
and for $K \geq 3$:
\begin{align} \label{eq:pz_Kgeq3}
  p_1 &= 1 - P_I(x+1), \nonumber \\
  p_z &= \prod_{j=1}^{z-1} P_I(x+j) (1-P_I(x+z)),z=2,\ldots,K-1,\nonumber \\
  p_K &= \prod_{j=1}^{K-1} P_I(x+j).
\end{align}

In the following lemma, we establish the structure of the distribution of $k_0$.
\begin{lemma} \label{lemma:dist_k0}
The distribution of $k_0$ is monotonically decreasing in the idle age, $x_t$, for $z = 2,\ldots, K$, and monotonically increasing in $x_t$, for $z=1$.
\end{lemma}
\begin{proof}
  For $K=2$, it is straightforward to establish the conclusion. We next focus on $K\geq 3$. First, it is easy to show that for $z=1$, $p_1$ is monotonically increasing in $x$ since $P_I(x)$ is monotonically decreasing. Further, note that if two positively valued functions, $f_1(x) > 0$ and $f_2(x) > 0$, are both monotonically decreasing in $x$,  i.e., for any positive integer $\Delta \geq 1$,
  \begin{equation*}
    f_1(x) > f_1(x + \Delta),~~~f_2(x) > f_2(x+\Delta),
  \end{equation*}
  then the product of the two, $f_1(x)f_2(x)$, is also a decreasing function since
  \begin{equation*}
    f_1(x) f_2(x) >  f_1(x + \Delta) f_2(x+\Delta).
  \end{equation*}
  Therefore, $p_K$ is monotonically decreasing in $x$.

  Next, we show that $p_z,z=2,\ldots,K-1$ are monotonically decreasing in $x$. Based on the above argument, it is sufficient to show that for any $z = 2,\ldots, K-1$, the following function,
  \begin{equation*}
    g(x) \triangleq P_I(x+z-1) (1-P_I(x+z)),
  \end{equation*}
  is monotonically decreasing.
  We have the following {simplifications}: For $\Delta \geq 1$, a positive integer,
  \begin{eqnarray}
   \frac{g(x)}{g(x+\Delta)} \hspace{-3mm}&=&\hspace{-3mm} \frac{P_I(x+z-1) (1-P_I(x+z))}{P_I(x+\Delta+z-1) (1-P_I(x+\Delta+z))} \nonumber \\
 \hspace{-3mm}&=&\hspace{-3mm} \frac{(x+\Delta+z-1)^u + C_I}{(x+z)^u + C_I} \hspace{-1mm}\cdot\hspace{-1mm} \frac{(x+\Delta+z)^u + C_I}{(x+\Delta+z)^u + C_I-1}
 \cdot \frac{(x+z)^u + C_I-1}{(x+z-1)^u + C_I}
 \nonumber \\
 &\triangleq& {B}_1 {B}_2 {B}_3.
  \end{eqnarray}
  Since $x \geq 0, C_I > 0, u\geq 1$ and $\Delta \geq 1$, we immediately have $B_1 \geq 1$, $B_2 > 1$. Further, using the binomial theorem, we obtain
  \begin{equation}
    (x+z)^u
    = \sum_{i=0}^u {u\choose i} (x+z-1)^i \geq (x+z-1)^u + 1,
  \end{equation}
  and hence $B_3 \geq 1$.
  As a result, $\frac{g(x)}{g(x+\Delta)} > 1$ and  $p_2, \ldots, p_{K-1}$ are monotonically decreasing in $x$. This concludes the proof.
\end{proof}

We next present a result that demonstrates the tradeoff inherent in the immediate reward with respect to age.

\begin{proposition} \label{prop:Rt_age}
  The immediate reward on the scheduled channel
  is monotonically decreasing in the idle age.
\end{proposition}
\begin{proof}
  As the system model implies, we can rewrite the immediate reward as the following weighted sum:
    \begin{equation} \label{eq:Rt}
      R_t(\textbf{S}_t,a_t) = \sum_{z = 1}^K \sum_{k=1}^z \pi_{t,k}^s(a_t)p_z,
    \end{equation}
    where $p_z =  \textrm{Pr}(k_0 = z)$ is given by (\ref{eq:pz_Kgeq3}).

    When $K=2$, with idle age on $a_t$ being $x_t$, we have $R_t(\textbf{S}_t,a_t) = \pi_{t,1}^s(a_t) + P_I(x_t+1) \pi_{t,2}^s(a_t)$. Apparently, it increases as $x_t$ decreases.

    For $K\geq 3$, denote by $\theta_z \triangleq \sum_{k=1}^z \pi_{t,k}^s(a_t), z = 1,\ldots,K$. It is clear that
    \begin{equation}
      \theta_1 < \theta_2 < \cdots < \theta_K,
    \end{equation}
    and $\{\theta_z\}$'s are constants in the idle age $x_t$. To emphasize the role of the argument $x_t$, we rewrite $R_t(\textbf{S}_t,a_t)$ as $R_t(x_t)$, and $p_z$ as $p_z(x_t)$.

    Utilizing the result of Lemma~\ref{lemma:dist_k0}, and noting that for any positive integer $\Delta \geq 1$, $|p_1(x_t) - p_1(x_t + \Delta)| = \sum_{z=2}^K (p_z(x_t) - p_z(x_t + \Delta))$, we obtain:
    \begin{equation*}
      R_t(x_t) -R_t(x_t+\Delta) =\sum_{z=1}^K \theta_z \left(p_z(x_t) - p_z(x_t + \Delta)\right)> \theta_1(p_1(x_t) - p_1(x_t + \Delta)) + \theta_2 \sum_{z=2}^K (p_z(x_t) - p_z(x_t + \Delta))> 0,
    \end{equation*}
    i.e., $R_t(x_t)$ is monotonically decreasing in the idle age $x_t$, and this establishes the proposition.
\end{proof}

The above result can be readily extended to the following corollary.
\begin{corollary} \label{corollary:Rt_age}
  When all the PU channels have equal fading state beliefs, the immediate reward
  is maximized by scheduling the SU to the channel with the lowest idle age.
\end{corollary}

Next, recall that the channel fading is modeled by a positively-correlated Markov chain. Hence, if $\pi_{t,1}^s(a_t) > \pi_{t,1}^{s'}(a_t)$, then the inequality $\pi_{t,k}^s(a_t) > \pi_{t,k}^{s'}(a_t)$ holds, for all $k =2,\ldots,K$. We present the following proposition without further proof.
\begin{proposition}\label{prop:Rt_highest_pis}
  The immediate reward on the scheduled channel
  is monotonically increasing in its fading state belief at the moment.
  Further, given equal idle ages across all PU channels, the immediate reward is maximized by scheduling the SU to the channel with the largest fading state belief value at the moment.
\end{proposition}

\subsection{Tradeoffs Inherent in Total Reward}
In this subsection, we illustrate the tradeoffs inherent in the total reward by examining a special case with two channels $N=2$ and number of mini-slots $K=1$. In particular, we show that under these conditions, a simple greedy scheduling policy is optimal. The greedy policy is formally defined as follows: In any control slot, the greedy decision maximizes the immediate reward, ignoring the future rewards, i.e.,
\begin{equation}
    \hat{a}_t = \max_{a_t\in{\mathcal{A}_t}} \{R_t(\textbf{S}_t,a_t)\}.
\end{equation}
We now formally record the result on greedy policy optimality in the following proposition.
\begin{proposition} \label{prop:K=1}
    The greedy policy is optimal when $K=1$ and $N=2$.
\end{proposition}
\begin{proof}
    To prove the proposition, we begin with the following induction hypothesis:

    \textit{Induction Hypothesis: With the length of the horizon denoted by $m, m \geq 2$, assume that greedy policy is optimal in all the control slots $t\in\{m-1, \ldots,1\}$.}

    The proof proceeds in two steps: First, we fix a sequence of scheduling decisions $\vec{l} := \{a_{m},\ldots,a_{t+1}\}$, and show that the expected immediate reward in control slot $t$, under the greedy policy, is independent of the scheduling decisions $\vec{l}$. Then, we provide induction based arguments to validate the induction hypothesis and hence establish that the greedy policy is optimal in all the control slots.

    Let $U_t^{(\vec{l})}$ denote the expected immediate reward in slot $t \in\{m-1,\ldots,1\}$, given the scheduling decisions $\vec{l}$. $U_t^{(\vec{l})}$ can be calculated as:
    \begin{equation} \label{eq:Ut_l}
      U_t^{(\vec{l})} = \sum_{\{o_t(1),o_t(2)\}} U_t^{(\vec{l})}(o_t(1),o_t(2)) \textrm{Pr}(o_t(1),o_t(2)),
    \end{equation}
    where $o_t(1)$ (likewise, $o_t(2)$) is the binary indicator of whether channel 1 (or 2) is idle ($o_t(1)=0$) or not ($o_t(1)=1$) in the $t$th control slot, and $\textrm{Pr}(o_t(1),o_t(2))$ denotes the joint probability of both channels' availability status (idle or busy).
    There exist four realizations of the vector $(o_t(1),o_t(2))$, namely $\{(0,0),(0,1),(1,0),(1,1)\}$.
    In what follows, we show that the value of $U_t^{(\vec{l})}(o_t(1),o_t(2))$ calculated under each of the realizations is independent of the scheduling decisions $\vec{l}$.

    \textit{Case 1: When $(o_t(1),o_t(2)) = (0,0)$.} In this case, both channels are idle in control slot $t$. Let $\pi_{t+1}^s(n)$ be the fading state belief on channel $n$ in control slot $t+1$. The expected immediate reward in control slot $t$, under the greedy policy, can be calculated as
    \begin{equation}
      U_t^{(\vec{l})}(0,0) \hspace{-1mm}=\hspace{-1mm} \pi_{t+1}^s(a_{t+1})p
                +(1-\pi_{t+1}^s(a_{t+1}))\mathrm{T}(\pi_{t+1}^s(\tilde{a}_{t+1})),
    \end{equation}
    where $\tilde{a}_{t+1} = \{1,2\} \backslash a_{t+1}$.
    For notational convenience, we write $\alpha \triangleq \pi_{t+1}^s(a_{t+1})$ and $\tilde{\alpha} \triangleq \pi_{t+1}^s(\tilde{a}_{t+1})$. The reward $U_t^{(\vec{l})}(0,0)$ can now be further expressed as
    \begin{eqnarray}
      U_t^{(\vec{l})}(0,0) &=& p \alpha + (1-\alpha) \mathrm{T}(\tilde{\alpha}) \nonumber \\
                &=&  p \alpha + (1-\alpha) (\alpha p + (1-\alpha)r) \nonumber \\
                &=& p P_1 + r P_2,
    \end{eqnarray}
    where
    \begin{equation}
        P_1 \triangleq \alpha + \tilde{\alpha} - \alpha \tilde{\alpha}, ~~~P_2 \triangleq (1-\alpha)(1-\tilde{\alpha}).
    \end{equation}
    That is, $P_1$ is the probability that at least one of the channels experiences good channel fading in the previous control slot $t+1$, while $P_2$ is the probability that both channels see bad channel fading. It is noted that these probabilities are controlled by the underlying Markov dynamics only, and thus $P_1$ and $P_2$ are independent of the scheduling decisions $\vec{l}$. Therefore, $U_t^{(\vec{l})}(0,0)$ is independent of $\vec{l}$.

    \textit{Case 2: When $(o_t(1),o_t(2)) = (0,1)$.} In this case, only channel 1 is idle and can be scheduled. The reward $U_t^{(\vec{l})}(0,1)$ is obtained as
    \begin{eqnarray}
      U_t^{(\vec{l})}(0,1) =
      \left\{
        \begin{array}{ll}
         p \pi_{t+1}^s(1) + r(1-\pi_{t+1}^s(1)), &\hspace{-2.5mm}\textrm{if}~ a_{t+1} = 1,\\
        \mathrm{T}(\pi_{t+1}^s(1)), &\hspace{-2.5mm}\textrm{if}~ a_{t+1} = 2.
      \end{array}
      \right. \hspace{-2mm}
    \end{eqnarray}
    It follows from~(\ref{eq:TL_alpha}) that
    \begin{equation}
      U_t^{(\vec{l})}(0,1)|_{a_{t+1} = 1} = U_t^{(\vec{l})}(0,1)|_{a_{t+1} = 2},
    \end{equation}
    i.e., $U_t^{(\vec{l})}(0,1)$ is independent of $\vec{l}$.

    \textit{Case 3: When $(o_t(1),o_t(2)) = (1,0)$.} Similar to \textit{Case 2}, only channel 2 can be scheduled in this case, and we have:
    \begin{eqnarray}
      U_t^{(\vec{l})}(1,0)=
      \left\{
        \begin{array}{ll}
        \mathrm{T}(\pi_{t+1}^s(2)), &\hspace{-2.5mm}\textrm{if}~ a_{t+1} = 1, \\
        p \pi_{t+1}^s(2) + r(1-\pi_{t+1}^s(2)), &\hspace{-2.5mm}\textrm{if}~ a_{t+1} = 2.
      \end{array}
      \right. \hspace{-2mm}
    \end{eqnarray}
    Again, this indicates that
    \begin{equation}
      U_t^{(\vec{l})}(1,0)|_{a_{t+1} = 1} = U_t^{(\vec{l})}(1,0)|_{a_{t+1} = 2},
    \end{equation}
    i.e., $U_t^{(\vec{l})}(1,0)$ is independent of $\vec{l}$.

    \textit{Case 4: When $(o_t(1),o_t(2)) = (1,1)$.} In this case, both channels are busy, and it follows immediately that
    \begin{equation}
      U_t^{(\vec{l})}(1,1) = 0.
    \end{equation}
    Clearly, $U_t^{(\vec{l})}(1,1)$ is independent of $\vec{l}$ as well.

    Next, note that $\textrm{Pr}(o_t(1),o_t(2))$ is a function of the ages $(x_t(1),x_t(2))$ only, which evolve independently from the scheduling decisions $\vec{l}$. Thereby, we conclude that $\textrm{Pr}(o_t(1),o_t(2))$ is independent of the scheduling decisions $\vec{l}$, and so is the expected immediate reward in control slot $t$, i.e, $U_t^{\vec{l}} = U_t$.
    Now, by extension, we have that the total reward collected from control slot $t$ till the horizon is independent of $\vec{l}$, i.e.,
    \begin{equation*}
      \sum_{t'= t}^1 U_{t'}^{\vec{l}} = \sum_{t'=t}^1 U_{t'}.
    \end{equation*}
    Thus, the greedy policy is optimal in control slot $t+1$ as well.
    Since $t \in \{m-1,\ldots\}$ is arbitrary, the greedy policy is optimal in every control slot $\{m,\ldots,1\}$ under the induction hypothesis. Finally, as the greedy policy is trivially optimal at the horizon, i.e., $t=1$, using backward induction, we validate the induction hypothesis, and arrive at the conclusion that greedy is optimal in all control slots $t \in \{m,\ldots,1\}$. This establishes the proposition.
\end{proof}

{\bf Remarks:} Note that the tradeoffs inherent in the special case considered above, i.e., $K=1, N=2$, is more intricate than those observed in related recent works~(e.g., \cite{myopic:Qing, sugumar}), where a control slot coincides with a mini-slot and only one of the states: channel fading or PU occupancy, is considered. This is because, despite $K=1$, the question of ``whether to learn a channel that may not be available for scheduling in the near future due to channel occupancy state'' still exists. Thus the tradeoffs discussed in the preceding subsections are retained in this special case, essentially adding value to our result on greedy optimality. In the subsequent section, we proceed to further understand the tradeoffs in the original system by introducing a conceptual ``genie-aided system.''

\section{Multi-tier Tradeoffs: A Closer Look via A Genie-Aided System} \label{sec:tradeoff_ga}
In the previous section, we partially showed the interaction between various state elements
by examining the immediate reward and certain special cases.
In order to obtain a more complete understanding of the inherent dynamics, we next introduce a full-observation genie-aided system that helps decompose and characterize the various tiers of the multi-dimensional tradeoffs.

\subsection{A Genie-Aided System}

The genie-aided system is a variant of the original system with the following modification: The spectrum server receives channel fading feedbacks from \textit{all} the channels and not only the scheduled channel. These feedbacks are collected at the same times as those of the feedback from the scheduled channel. Thus when the PU returns on the scheduled channel, the feedback from \textit{all} the channels stop at once. Note that this is a conceptual system, without practical significance, which as we will see, is helpful in better understanding the complicated tradeoffs inherent in the original system.

\subsection{Tradeoffs Associated with Channel Fading}

\begin{proposition} \label{prop:ga_sameage}
  When the idle ages are the same across all PU channels,
  it is optimal to schedule the SU to the channel with the highest fading state belief at the moment, i.e.,
  \begin{equation}
    a_t^* = \arg \max_n \{\pi_{t,1}^s(n)\}.
  \end{equation}
\end{proposition}
\begin{proof}
    First, from Proposition~\ref{prop:Rt_highest_pis}, the immediate reward is maximized when scheduling the channel with the highest fading state belief at the moment. Now, we focus on showing that the future reward is independent of the action in the current control slot and therefore establishing the proposition. Specifically,
    at the current control slot $t$, $t \geq 2$, consider an arbitrary control slot in the future, $t_0 \in \{t-1,\ldots,1\}$. In the following, we show that the expected immediate reward in this control slot, denoted by $U_{t_0}$, is independent of the current action $a_t$, and thus the future reward is independent of $a_t$.

    Let $k_0'$ denote the latest idle mini-slot in control slot $t_0 + 1$ before the PU returns. The reward $U_{t_0}$ is then calculated as:
    \begin{equation} \label{eq:Ut_Prk0}
        U_{t_0} = \sum_{k_0' = 1}^K U_{t_0}(k_0') \textrm{Pr}(k_0'),
    \end{equation}
    where $\textrm{Pr}(k_0')$ is the distribution of $k_0'$, identical to that of $k_0$ as given in~(\ref{eq:pz_Kgeq3}), and $U_{t_0}(k_0')$ is the expected reward in control slot $t_0$ for a given $k_0'$.

    In the genie-aided system, the spectrum server obtains the feedbacks of the fading states from \textit{all} $N$ channels simultaneously, i.e., at the end of mini-slot $k_0'$. Place the channels on which good channel fading is observed at $k_0'$ in the set $\mathcal{C}_1$, and the rest in another set $\mathcal{C}_0$.
    The characterization of
    the reward $U_{t_0}(k_0')$ can be further divided into the following cases.
    \begin{itemize}
      \item \textit{Case~1: $\mathcal{C}_0 = \emptyset$.}~~This corresponds to the case where $f_{t_0+1,k_0'}(n) = 1, \forall n = 1,\ldots,N$.
        Let $W_{t_0}^{(\textrm{case} 1)}(n)$ be the expected reward in control slot $t_0$ on channel $n$ in this case.
        We have
          \begin{equation} \label{eq:Wt_case1}
            W_{t_0}^{(\textrm{case} 1)}(n) = \pi_{t_0,1}^o(n) \sum_{z=1}^K \sum_{k=1}^z \pi_{t_0,k}^s(n) p_z(x_{t_0}(n)),
          \end{equation}
          where $\pi_{t_0,1}^s(n) = \mathrm{T}^{K-k_0'}(p), \forall n \in\mathcal{C}_1$. It follows that $\pi_{t_0,k}^s(1) = \ldots = \pi_{t_0,k}^s(N), \forall k = 1,\ldots,K$. Further, since $x_t(1) = \ldots =x_t(N)$, we have $p_z(x_{t_0}(1)) = \ldots = p_z(x_{t_0}(N))$, and $\pi_{t_0,1}^o(1) = \ldots = \pi_{t_0,1}^o(N)$, and therefore, scheduling \textit{any} of the idle channels in $\mathcal{C}_1$ achieves the same expected reward in control slot $t_0$, i.e.,
          \begin{equation}
            U_{t_0}(k_0')|_{\textrm{case} 1} = W_{t_0}^{(\textrm{case} 1)}(n), \forall n \in \mathcal{C}_1, o_{t_0,1}(n)=0.
          \end{equation}
          Now, note that given equal idle ages on all $N$ channels, the distribution of $k_0'$ is identical across channels. Therefore, for all $n=1,\ldots,N$, the values of $\pi_{t_0,k}^s(n), k=1,\ldots,K$, $\pi_{t_0,1}^o(n)$ and $p_z(x_{t_0}(n))$ all stay unchanged if a different channel $a_t \neq a_t'$ is scheduled in the current control slot $t$. This implies that $U_{t_0}(k_0')|_{\textrm{case} 1}$ is unchanged, and is thus independent of $a_t$.

      \item \textit{Case~2: $\mathcal{C}_1 = \emptyset$.}~~In this case, $f_{t_0+1,k_0'}(n) = 0, \forall n = 1,\ldots,N$. The expected reward collected in control slot $t_0$ on channel $n$, denoted by $W_{t_0}^{(\textrm{case} 2)}(n)$, can be expressed the same as~(\ref{eq:Wt_case1}), where the channel strength belief $\pi_{t_0,1}^s(n) = \mathrm{T}^{K-k_0'}(r)$, for all $n\in\mathcal{C}_0$. Then, using the similar reasoning as in \textit{Case 1}, we obtain that scheduling \textit{any} of the idle channels in $\mathcal{C}_0$ achieves the same expected reward in control slot $t_0$, i.e.,
          \begin{equation}
            U_{t_0}(k_0')|_{\textrm{case} 2} = W_{t_0}^{(\textrm{case} 2)}(n), \forall n \in \mathcal{C}_0, o_{t_0,1}(n)=0.
          \end{equation}
           Further, the expected reward achieved in this case, $U_{t_0}(k_0')|_{\textrm{case} 2}$, does not change when the action in current control slot $t$ varies, i.e., $U_{t_0}(k_0')|_{\textrm{case} 2}$ is independent of $a_t$.

      \item \textit{Case~3: $\mathcal{C}_0 \neq \emptyset, \mathcal{C}_1 \neq \emptyset$.}~~
         In this case, we first show that the maximum expected reward in control slot $t_0$ is achieved by scheduling {any} of the idle channels in the set $\mathcal{C}_1$, which are perceived to have better channel fading state in the subsequent control slot $t_0$ than the channels in set $\mathcal{C}_0$. Specifically,
         picking any one of the channels from each of the set, $n_1 \in \mathcal{C}_1$ and $n_0 \in \mathcal{C}_0$, we have
        \begin{eqnarray}
            W_{t_0}^{(\textrm{case} 3)}(n_1) \hspace{-2mm} &=&\hspace{-2mm} \sum_{z=1}^K \sum_{k=1}^z \pi_{t_0,k}^s(n_1) p_z(x_{t_0}(n_1)), \nonumber \\
            W_{t_0}^{(\textrm{case} 3)}(n_0) \hspace{-2mm} &=&\hspace{-2mm} \sum_{z=1}^K \sum_{k=1}^z \pi_{t_0,k}^s(n_0) p_z(x_{t_0}(n_0)),
        \end{eqnarray}
        where
        \begin{eqnarray}
            \pi_{t_0,k}^s(n_1) &=& \mathrm{T}^{K-k_0'+k-1}(p),\nonumber \\
            \pi_{t_0,k}^s(n_0) &=& \mathrm{T}^{K-k_0'+k-1}(r).
        \end{eqnarray}
         Based on~(\ref{eq:TL_alpha}) and the property of the positively-correlated Markov chain, the following inequality holds: For all $k = 1,\ldots,K$,
         \begin{equation*}
           \pi_{t_0',k}^s(n_1) = \mathrm{T}^{K-k_0'+k-1}(p) \geq \mathrm{T}^{K-k_0'+k-1}(r) = \pi_{t_0,k}^s(n_0),
         \end{equation*}
         with the equality achieved when $K \rightarrow \infty$.
         Further, applying similar argument as before, we obtain
          \begin{equation*}
            W_{t_0}^{(\textrm{case} 3)}(n_1) >  W_{t_0}^{(\textrm{case} 3)}(n_0).
          \end{equation*}
        Now, since $n_1 \in \mathcal{C}_1$ (likewise, $n_0  \in \mathcal{C}_0$) is arbitrary, from the conclusion drawn in \textit{Case 1}, the maximal expected reward in control slot $t$ under this case is achieved by scheduling the SU to \textit{any} of the idle channels in the set $\mathcal{C}_1$, i.e.,
        \begin{equation}
          U_{t_0}(k_0')|_{\textrm{case} 3} = W_{t_0}^{(\textrm{case} 3)}(n_1), \forall n_1 \in \mathcal{C}_1, o_{t_0,1}(n_1)=0.
        \end{equation}
        Next, based on the similar reasoning as in the previous cases, we know that $U_{t_0}(k_0')|_{\textrm{case} 3}$ is independent of $a_t$.
    \end{itemize}

    Finally, note that $U_{t_0}(k_0')$ can be written as:
    \begin{equation}
        U_{t_0}(k_0') = \sum_{i=1}^3 U_{t_0}(k_0')|_{\textrm{case}~i} \textrm{Pr} (\textrm{Case}~ i),
    \end{equation}
    where
    \begin{eqnarray*} \label{eq:prob_case123}
    \textrm{Pr} (\textrm{Case} 1) &=& \prod_{n \in \mathcal{C}_1} \pi_{t_0+1,k_0'}^s(n), \nonumber \\
     \textrm{Pr} (\textrm{Case} 2) &=& \prod_{n \in \mathcal{C}_0} (1-\pi_{t_0+1,k_0'}^s(n)), \nonumber \\
     \textrm{Pr} (\textrm{Case} 3) &=& \prod_{n \in \mathcal{C}_1} \pi_{t_0+1,k_0'}^s(n) \prod_{n'\in \mathcal{C}_0} (1-\pi_{t_0+1,k_0'}^s(n')),
    \end{eqnarray*}
    are quantities dependent on $k_0'$ only.
    Based on the fact that the idle age at the moment are identical across the channels, the probabilities $\textrm{Pr} (\textrm{Case}~i), i=1,2,3$, and
    the distribution $\textrm{Pr}(k_0')$, are the same across the channels as well. Thus $U_{t_0}$ is independent of $a_t$.

    Since $t_0 \in \{t-1,\ldots,1\}$ is arbitrary, by extension, we have that the total reward collected from control slot $t_0$ till the horizon, i.e., $\sum_{t' = t_0}^1 U_{t'}$, which is the future reward of current control slot $t$, is independent from $a_t$. This concludes the proof and establishes the proposition.
\end{proof}

\textbf{Remarks:} Proposition~\ref{prop:ga_sameage} illustrates the effect of fading state belief on the optimal decisions in the genie-aided system. With ages equalized across the PU channels and the classic ``exploitation vs. exploration'' tradeoff neutralized (by definition of the genie-aided system), we saw that, higher fading belief favors the immediate reward and that the future reward is, in fact, independent of the current decision.

In the following, we study the effect of PU occupancy and age on the optimal decisions in the genie-aided system and, in turn, its impact on the original system.

\subsection{Tradeoffs Associated with PU Occupancy}

The following proposition identifies the effect of channel occupancy state on the optimal scheduling decisions.
\begin{proposition} \label{prop:ga_samestrength}
  When the fading state beliefs are the same across all PU channels,
  it is optimal to schedule the SU to the channel with the lowest idle age at the moment, i.e.,
  \begin{equation}
    a_t^* = \arg \min_n \{x_t(n)\}.
  \end{equation}
\end{proposition}

\begin{proof}
  We prove the proposition by showing that scheduling the channel with the lowest idle age favors:
  \begin{inparaenum}[\itshape 1\upshape)]
    \item the immediate reward $R_t(\textbf{S}_t,a_t)$; and
    \item the optimal future reward $V_{t-1}^*(\textbf{S}_{t-1})$.
\end{inparaenum}
  The first part can be readily shown by appealing to Proposition~\ref{prop:Rt_age} and Corollary~\ref{corollary:Rt_age}. To show the second part, we adopt the following induction hypothesis:

  \textit{Induction Hypothesis: The optimal future reward is convex in the fading state belief.}

  When channel $a_t$ is scheduled in the current control slot $t$, the expected future reward can be evaluated as:
  \begin{equation*}
    V_{t-1}^*(\mathbf{S}_{t-1})|_{a_t}
                =\pi_{t,k_0}^s(a_t) V_{t-1}^*(p) + (1-\pi_{t,k_0}^s(a_t))V_{t-1}^*(r),
  \end{equation*}
  where $V_{t-1}^*(p)$ and $V_{t-1}^*(r)$ represent the future reward calculated when the channel fading state observed in the $k_0$th mini-slot of control slot $t$ is good or bad, respectively. More specifically,
  \begin{eqnarray*}
    V_{t-1}^*(p) \triangleq V_{t-1}^*\left(\pi_{t-1,1}^s(a_t) = \mathrm{T}^{K-k_0}(p)\right), \nonumber \\
    V_{t-1}^*(r) \triangleq V_{t-1}^*\left(\pi_{t-1,1}^s(a_t) = \mathrm{T}^{K-k_0}(r)\right).
  \end{eqnarray*}

  Based on~(\ref{eq:TL_alpha}), we have, for $\gamma\in(0,1)$,
  \begin{equation*}
    \mathrm{T}^L(\gamma) = (p-r)^L\gamma + r\frac{1-(p-r)^L}{1-(p-r)}, L = 0,1,\ldots,
  \end{equation*}
  and $\lim_{L \rightarrow \infty} \mathrm{T}^L(\gamma) = \frac{r}{1-p+r} \triangleq \pi_{s|s}$, where $\pi_{s|s}$ denotes the steady-state probability of perceiving good channel fading on any of the PU channels. This indicates that, a smaller $k_0$, associated with a higher idle age at the moment (recall discussions in Section~\ref{sec:tradeoffs}), results in a larger $K-k_0$ and thus a value of $\pi_{t-1,1}^s(a_t)$ closer to $\pi_{s|s}$, in which case,
  \begin{eqnarray*}
    V_{t-1}^*(\mathbf{S}_{t-1})|_{a_t} \hspace{-3.5mm}&\rightarrow&\hspace{-3.5mm} \pi_{t,k_0}^s(a_t) V_{t-1}^*(\pi_{s|s}) \hspace{-1.2mm}+\hspace{-1.2mm} (1-\pi_{t,k_0}^s(a_t))V_{t-1}^*(\pi_{s|s})
    \triangleq V_{t-1}^*(a_t,E_{\pi_{s|s}}).
  \end{eqnarray*}
  On the contrary, as $k_0$ becomes larger because of a lower idle age, the value of $\pi_{t-1,1}^s(a_t)$ deviates further away from $\pi_{s|s}$, but is closer to $p$ (or $r$). Also, $\pi_{t,k_0}^s(a_t)$ gets closer to $\pi_{s|s}$. Therefore,
  \begin{eqnarray*}
    V_{t-1}^*(\mathbf{S}_{t-1})|_{a_t} \hspace{-3.5mm}&\rightarrow&\hspace{-3.5mm} \pi_{s|s} V_{t-1}^*(p) \hspace{-1.1mm}+\hspace{-1.1mm} (1-\pi_{s|s}) V_{t-1}^*(r) \triangleq E_{\pi_{s|s}} V_t^*(a_t).
  \end{eqnarray*}

  Now, appealing to the induction hypothesis and the Jensen's inequality~\cite{Boyd}, we have that $E_{\pi_{s|s}} V_t^*(a_t) > V_{t-1}^*(a_t,E_{\pi_{s|s}})$, and therefore the future reward is maximized by scheduling the channel with the lowest idle age, i.e.,
  \begin{equation*}
    a_t^* = \min_{a_t}\{x_t(a_t)\}~~s.t.~~ V_{t-1}^*(\textbf{S}_{t-1})|_{a_t^*} = \max_{a_t}\{V_{t-1}^*(\textbf{S}_{t-1})|_{a_t}\}.
  \end{equation*}

  Finally, we proceed to verify the induction hypothesis. At $t=2$, the optimal future reward equals the optimal immediate reward at the horizon $t=1$, i.e.,
  \begin{equation}
    V_1^*(\textbf{S}_1) = R_1(\textbf{S}_1, a_1^*) := \max_{a_1\in\mathcal{A}_1}\{R_1(\textbf{S}_1,a_1)\}.
  \end{equation}
  Since $R_1(\textbf{S}_1,a_1)$ is linear in the strength belief, using the property of convex function~\cite{Boyd}, we have that $V_1^*(\textbf{S}_1)$ is convex in the strength beliefs, which validates the induction hypothesis. Then, using backward induction, we establish the proposition.
\end{proof}

\textbf{Remarks:} Proposition~\ref{prop:ga_samestrength} explicitly illustrates the effect of the PU occupancy and age on the optimal decisions. With fading state beliefs equalized and the classic ``exploitation vs. exploration'' tradeoff neutralized (by definition of the genie-aided system), we saw that:
\begin{inparaenum}[\itshape 1\upshape)]
\item lower idle age on the PU channel favors the immediate reward by allowing more idle time on the channel; and
\item lower idle age also favors the future reward by way of better freshness of the channel fading feedback.
\end{inparaenum}

Thus by studying the full-observation genie-aided system, via the results in Propositions~\ref{prop:ga_sameage} and~\ref{prop:ga_samestrength}, we have decomposed the tradeoffs associated with the channel occupancy and the fading state beliefs in the original system. Indeed, the results in Propositions~\ref{prop:ga_sameage} and~\ref{prop:ga_samestrength} rigorously support the understanding we developed earlier in Section~\ref{sec:tradeoffs} on the tradeoffs in the original system.

\section{Numerical Results \& Further Discussions}\label{sec:numresult}

In this section, we evaluate and compare the optimal rewards of the original system (denoted as $V_{ori}^*$) and the genie-aided system ($V_{genie}^*$). Also, the optimal policy in the original system is compared to a \textit{randomized} scheduling policy (with reward denoted as $V_{random}^*$), where the spectrum server chooses a channel, among all the idle ones, randomly and uniformly, and allocates it to the SU for data transmission. The numerical results are collected for a two channel system with $K=2$, horizon length $m=6$, and discounted factor $\beta = 0.9$. For notational convenience, denote $\Delta_{{ga}-{ori}} = V_{genie}^* - V_{ori}^*$ and $\Delta_{{ori}-{rnd}} = V_{ori}^* - V_{random}^*$.

Table~\ref{table:pr_delta} records the rewards obtained under various baseline cases, for various values of $\delta \triangleq p-r$, which broadly captures the temporal memory in the channel fading. In particular, as $\delta$ decreases, the channel fading memory fades and the difference between the baselines, which are primarily differentiated by the degree to which they exploit the memory in the system, tends to decrease. This is observed from Table~\ref{table:pr_delta}.
\begin{table}[tbh]
\caption{Comparison of rewards when the channel fading memory varies. System parameters used: $u=1,C_I = 1, C_B = 2, x_t(1) = 10, x_t(2) = 5, \pi_{t,1}^s(1) = 0.4, \pi_{t,1}^s(2) = 0.7$}
    \begin{center}
    \label{table:pr_delta}
    \begin{tabular}{ | l | l | l | l | l |}
      \hline
      $\delta_{pr}$ & $ \Delta_{{ga}-{ori}}$ & $\frac{ \Delta_{{ga}-{ori}}}{V_{genie}^*}\%$ & $\Delta_{{ori}-{rnd}}$ & $\frac{\Delta_{{ori}-{rnd}}}{V_{ori}^*}\%$\\
      \hline
        0.8   & 0.0077 & 0.3\% & 0.372  & 14.5\%  \\
      \hline
        0.4   & 0.0068 & 0.29\% & 0.2827 & 9.42\%     \\
      \hline
        0.2   & 0.0023 & $9.8\times 10^{-4}\%$ & 0.1915 & 8.23\%     \\
      \hline
        0.1   & 0.0006 & $2.4\times 10^{-4}\%$ & 0.1836 & 7.93\%     \\
      \hline
    \end{tabular}
    \end{center}
\end{table}

Next, in Table~\ref{table:u}, we compare the baseline rewards under different values of the power exponent $u$, used in the definitions of $P_I$ and $P_B$ in (\ref{eq:PI_PB}).
To build a better understanding of the trend reflected in these numerical results, we consider an arbitrary mini-slot, denoted as $k'$, as the current mini-slot, and the following two exhaustive sets of PU occupancy histories:
\begin{inparaenum}[\itshape 1\upshape)]
\item the set of histories $h_x^I, x=1,2,\ldots$, corresponds to the case when there are exactly $x$ mini-slots between the current mini-slot $k'$ and the most recent mini-slot (preceding $k'$) when the channel was in busy state,
    i.e., the idle age of mini-slot $k'$ is $x$; and
\item the set of histories $h_x^B, x=1,2,\ldots$, corresponds to the case when there are exactly $x$ mini-slots between the current mini-slot $k'$ and the most recent mini-slot (preceding $k'$) when the channel was in idle state,
    i.e., the busy age of mini-slot $k'$ is $x$.
\end{inparaenum}
In Fig.~\ref{fig:occp_history_h}, we plot the two sets of occupancy histories. As has been pointed out in Section~\ref{subsec:probformulation}, the idle/busy age is a sufficient statistic for capturing the memory in the PU channels' occupancy states. Thereby, with the construction of $h_x^I$ and $h_x^B$, we can examine the effect of the temporal correlation of PU occupancy on the system performance. Specifically, the conditional idle probability in mini-slot $k'$, given occupancy history, can be obtained as:
\begin{eqnarray}
  \pi_{k'}^o |_{h_x^I} \hspace{-2mm}&=&\hspace{-2mm} P_I(x) = \frac{1}{x^u+C_I},\nonumber \\
  \pi_{k'}^o |_{h_x^B} \hspace{-2mm}&=&\hspace{-2mm} 1 - P_B(x) = 1 - \frac{1}{x^u+C_B}.
\end{eqnarray}
As an example, we plot $\pi_{k'}^o |_{h_x^I}$ in Fig.~\ref{fig:occpmemory_idle}. It is clear that as $u$ increases, the conditional probability curves become steeper. Define the threshold point, $x_0$, such that for all $x>x_0$, the difference in $\pi_{k'}^o |_{h_x^I}$ is insignificant ({below $10^{-2}$}).
Now, note from the figure that the threshold $x_0$ decreases with increasing $u$, i.e., $x_0^{(u=5)} < x_0^{(u=3)} < x_0^{(u=1)}$. This indicates that the impact of different occupancy histories on the current PU occupancy state diminishes with increasing $u$ and thus a decreased memory in the PU occupancy. Similar argument holds when considering $h_x^B$. In short, the exponent $u$ broadly captures the PU occupancy memory, and as its value increases, the memory decreases and thus the rewards under various cases, as expected, come closer with increasing $u$. This is illustrated in Table~\ref{table:u}.

\begin{figure}
  \centering
  \subfigure[]{\label{fig:occp_history_h}\includegraphics [width = 0.38\textwidth] {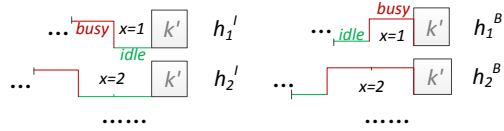}}\\
  \subfigure[]{\label{fig:occpmemory_idle}\includegraphics [width = 0.41\textwidth] {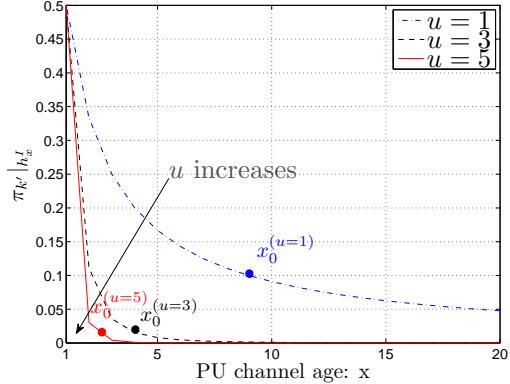}}
  \caption{Occupancy histories \& the conditional idle probability.}

\end{figure}

\begin{table}[tbh]
\caption{Comparison of rewards when the occupancy memory varies. System parameters used: $p = 0.9, r = 0.1, C_I = 1, C_B = 2, x_t(1) = 0, x_t(2) = 1, \pi_{t,1}^s(1) = 0.4, \pi_{t,1}^s(2) = 0.7$}
    \begin{center}
    \label{table:u}
    \begin{tabular}{ | l | l | l | l | l |}
      \hline
      $u$ & $ \Delta_{{ga}-{ori}}$ & $\frac{ \Delta_{{ga}-{ori}}}{V_{genie}^*}\%$ & $\Delta_{{ori}-{rnd}}$ & $\frac{\Delta_{{ori}-{rnd}}}{V_{ori}^*}\%$\\
      \hline
        1   & 0.0088 & 0.35\% & 0.3273 &  12.66\%   \\
      \hline
        3   & 0.006 & 0.26\% & 0.2201 & 9.59\%     \\
      \hline
        5   & 0.0043 & 0.2\% & 0.1839 & 8.28\%     \\
      \hline
    \end{tabular}
    \end{center}
\end{table}

Finally, as can be seen from both tables, the original system performs very close to the genie-aided system, while the cost of measuring and sending the channel fading feedback is only $\frac{1}{N}$ of the latter.
Also, the optimal policy significantly outperforms the randomized policy, indicating that the temporal correlation structure in the channel fading and PU occupancy can greatly benefit the opportunistic spectrum scheduling, when appropriately exploited.

\section{Conclusions} \label{sec:conclusion}
In this work, we studied opportunistic spectrum access for a single SU in a CR network with multiple PU channels. We formulated the problem as a partially observable Markov decision process, and examined the intricate tradeoffs in the optimal scheduling process, when incorporating the temporal correlation in \textit{both} the channel fading and PU occupancy states. We modeled the channel fading variation with a two-state first-order Markov chain. The temporal correlation of PU traffic was modeled using age and a class of monotonically decreasing functions, which can have a long memory. The optimality of the simple greedy policy was established under certain conditions. For the general case, we individually studied the tradeoffs in the immediate reward, and the total reward. Further, by developing a genie-aided system with full observation of the channel fading feedbacks, we decomposed and characterized the multiple tiers of the intricate tradeoffs in the original system. Finally, we numerically studied the performance of the two systems and showed that the original system achieved an optimal total reward very close (within $1\%$) to that of the genie-aided system. Further, the optimal policy in the original system significantly outperformed randomized scheduling, highlighting the merit of exploiting the temporal correlation in the system states. We believe that our formulation and the insights we have obtained open up new horizons in better understanding spectrum allocation in cognitive radio networks, with problem settings that go beyond the traditional ones.

\end{document}